\documentclass[conference,a4paper]{IEEEtran}
\usepackage{amssymb,amsmath,amsthm}
\usepackage{graphics}
\usepackage{subfigure} 

\theoremstyle{plain}
\theoremstyle{plain}
\theoremstyle{plain}
\theoremstyle{plain}
\newtheorem{prop}{Proposition}
\theoremstyle{plain}
\newtheorem{lemma}{Lemma}
\theoremstyle{plain}
\theoremstyle{plain}
\newtheorem{conjecture}{Conjecture}
\theoremstyle{remark}
\newtheorem{remark}{Remark}
\theoremstyle{discussion}
\theoremstyle{plain}

\theoremstyle{plain}
\newtheorem{thm}{Theorem}

\newtheorem{definition}{Definition}

\makeatletter

\newcommand{\Rmnum}[1]{\expandafter\@slowromancap\romannumeral #1@}

\makeatother

\def\F{{\mathbb F}}

\def\1{{\mathds 1}}

\newcommand{\pp}{\mathbb{P}}

\newcommand{\e}{\varepsilon}

\begin{document}

\sloppy
 
\title{A Shannon Approach to Secure Multi-party Computations} 

\author{
  \IEEEauthorblockN{Eun Jee Lee and Emmanuel Abbe}
  \IEEEauthorblockA{Program in Applied and Computational Mathematics,
    Princeton University\\
    Email: ejlee, eabbe@princeton.edu}
}

\maketitle

%
%
\begin{abstract} 
In secure multi-party computations (SMC), parties wish to compute a function on their private data without revealing more information about their data than what the function reveals. In this paper, we investigate two Shannon-type questions on this problem. We first consider the traditional one-shot model for SMC which does not assume a probabilistic prior on the data. In this model, private communication and randomness are the key enablers to secure computing, and we investigate a notion of randomness cost and capacity. We then move to a probabilistic model for the data, and propose a Shannon model for discrete memoryless SMC. In this model, correlations among data are the key enablers for secure computing, and we investigate a notion of dependency which permits the secure computation of a function. 
While the models and questions are general, this paper focuses on summation functions, and relies on polar code constructions.
\end{abstract}

%
%
\section{Introduction}
Consider a group of $m$ parties, each with a private bit $x_i$, $i \in [m]=\{1,\dots,m\}$, which are interested in computing jointly a function $f(x_1,\dots,x_m)$, without revealing any other information (than what the function reveals) about their inputs to anybody else. For example, 
the parties want to vote between two candidates for presidency without revealing their vote, i.e., $f(x_1,\dots,x_m)=\sum_{i=1}^n x_i$. Can this be achieved? 

Note that we are asking here for an exact computation of the function $f$, with an arbitrary number of parties $m$ (possibly low), and no information leakage. The latter requirement means that no additional information about the inputs must be shared than what would be shared in a model with a trusted party, which takes care of the computation. 
Hence, even in the case of a summation function, a noise-perturbation approach will not work in this framework. 
Of course, the above cannot be achieved without leveraging some ``security primitive". 
With secure multi-party computations (SMC), this goal is achieved by assuming that the parties have access to private communication. 
The ideas of SMC were first introduced by Yao in \cite{yao}, in a two-party setup, in particular with the millionaire problem.  
General multiparty protocols were then obtained by  
Goldreich, Micali and Wigderson \cite{gmw} for computational security, and by Ben-Or, Goldwasser and Wigderson \cite{bgw} and by Chaum, Cr\'epeau and Damg\.ard \cite{ccd} for information-theoretic security, using in particular secret-sharing \cite{shamir}. This paper focuses on the latter setting. 

Information-theoretic (IT) security does not rely on the computational power of the adversary, i.e., on hardness assumptions. The models and questions in IT SMC are however very different than the ones studied in the Shannon information theory models. In this paper, we consider the following problems. First, we consider a traditional model for SMC (with private communication, private access to randomness and honest-but-curious parties) and investigate a notion of randomness cost needed to compute a given function securely. Identifying the least amount of randomness is primarily a question which we find mathematically interesting and which connects to information theory subjects, in particular to the study of entropy vectors. It is however also a notion which captures the complexity of a function $f$ for its secure computation\footnote{There are various other complexity measures, such as the number of communication rounds and the computational complexity.}. 
In the second part of the paper, we propose a Shannon model for SMC, assuming the parties input to be drawn from a discrete memoryless source, and requiring the function computation and the security requirement to hold up to a vanishing error probability in an asymptotic regime. In this model, the correlation among the data can be leveraged to obtain secure computations. 
 This model departs significantly from the traditional SMC models discussed above, on the other hand, it is defined in a similar setting as for traditional information theory problems such secrecy \cite{AC,maurer,CN} or wire-tap channels \cite{wiretap}. Along these lines, a Shannon type model was recently proposed in \cite{tyagi} for a notion of ``secure computation'', which is however different than SMC and the notions discussed in this paper. In \cite{tyagi}, the parties wish to compute a function on their inputs using communication links which are eavesdropped, and the goal is to compute $f$ without allowing an eavesdropper to compute it. This is different from our setting, where the communication links between parties are secured, and where the parties themselves are the eavesdropper toward one another. Other works relevant to our setting are the interactive source compression \cite{ishwar} and the compress and compute problems \cite{korner-marton,code-comp}, but again, these do not take into account the privacy of the inputs among the parties.

%
%
\section{Notation}
In what follows, $[n]=\{1,...,n\}$, $|A|$ denotes the cardinality of a set $A$, $X^n$ denotes a vector of length $n$ and $X_i$ represents the $i^{th}$ element of the vector $X^n$. For $x\in \{0,1\}^n$ and $S\in [n]$, $x[S]=\{x_i:i\in S\}$. Finally, for two vectors $X^n$ and $Y^n$, $X^n\oplus Y^n$ represents the component-wise XOR addition.

%
%
\section{One-shot Model}
In this section, secure multi-party computation protocols are studied in the one-shot setting, where the parties' inputs have no probabilistic prior (equivalently a uniform prior) and where the function computation is done once. 


\begin{definition}
In the honest-but-curious-network (HCN) model,  
\begin{enumerate}
\item any pair of parties can communicate on a secured channel,
\item each party has access to randomness privately,
\item every party is honest-but-curious\footnote{Formal definitions are provided in \cite{goldreich}.}, i.e., the parties follow the protocol without deviating from it. However, collecting all the information exchanged in the protocol, the parties may try to learn additional information about other parties' inputs,
\item the parties have access to a synchronized clock.
\end{enumerate}
A protocol $\Pi$ in this model is a predetermined sequence of actions taken by the parties on a finite time scale $T$, where $T \geq 1$ is odd.   
At time $t=0$, each party possesses its own input. At an odd time $t \in \{1,3,\dots, T\}$, each party $P_i$ possesses $Y_{i,t}(\Pi)$, 
and can take the following actions:
\begin{enumerate}
\item[a)] draw a discrete random number $R_{i,t}(\Pi)$,
\item[b)] make a computation using $Y_{i,t}(\Pi)$ and $R_{i,t}(\Pi)$,
\end{enumerate}
and at even time $t \in \{2,4,\dots,T-1\}$, each party can transmit information to some other parties.  
Finally, we define the view of party $P_i$ from the protocol by $Y_i(\Pi)
=(Y_{i,1}(\Pi), R_{i,1}(\Pi), Y_{i,3}(\Pi), R_{i,3}(\Pi), \dots,Y_{i,T}(\Pi),R_{i,T}(\Pi))$.
\end{definition}

In SMCs, the notion of security is defined with the ideal vs. real model paradigm. In the ideal model, trusted parties provide securely their inputs to a trusted party which provides back securely the outputs to them. The real protocol is then deemed secure if any adversarial attack to the protocol has a counterpart that can be simulated in the ideal model. We refer to \cite{goldreich} for formal definitions.
In this paper, we adopt the following information theoretic definitions. 


\begin{definition}
Let $P_1,\dots,P_m$ denote $m$ parties, where party $P_i$ has input $X_i$, $i \in [m]$.  For simplicity, the inputs are assumed to take values in $\F_2=\{0,1\}$. We assume that $X_1,\dots,X_m$ are independent and uniformly distributed\footnote{One can define the inputs to be deterministic instead of uniform and provide worst-case notions of accuracy and security, which lead to the same result for the purpose of this paper.} random variables. 
A protocol $\Pi$ in the HCN model computes accurately and securely deterministic functions $f_1,\dots,f_m$ (taking values in a discrete set) if
\begin{itemize}
\item {\bf [Accuracy]} Each party $P_i$ can compute $f_i(X_1,\dots,X_m)$, i.e., $H(f_i(X_1,...,X_m)|Y_i(\Pi),X_i)=0$,
\item {\bf [Security]} Each party $P_i$ does not learn more information about the other parties' inputs than what the function reveals, i.e., $H(X_{\sim i}|X_i,Y_{i}(\Pi),f_i(X_1,...,X_m))$\\${\ \ \ \ \ \ \ \ \ \ \ \ \ \ \ =H(X_{\sim i}|X_i,f_i(X_1,...,X_m))}$,
\\where ${X_{\sim i}=(X_j : j\in [m] \setminus \{i\})}$. 
\end{itemize}
\end{definition}

%

\begin{definition}
We define the least amount of randomness that is required to compute accurately and securely functions $f_1,\dots,f_m$ on $m$ inputs $X_1,\dots,X_m$ by  
\begin{eqnarray}
\rho (f_1,\dots,f_m) = \inf_{\Pi} H(Y_1(\Pi),\dots,Y_m(\Pi) | X_1,\dots,X_m),
\end{eqnarray}
where the minimization is over all protocols $\Pi$ in the HCN model which computes accurately and securely $f_1,\dots,f_m$. 
\end{definition}
Notice that $\rho$ is invariant under the ordering of the functions. Note that all the $f_i$ may be the same. 
\subsection{The XOR function}
In this section, we investigate the XOR function, and assume that one of the party computes the function. It can always send back the output to other parties in a final round of the protocol. The techniques are based on traditional one-time pad and secret sharing steps. 
\begin{prop}
For any $m\geq 3$, 
\begin{eqnarray}
\rho (\emptyset,\dots,\emptyset, X_1\oplus \cdots \oplus X_m)= 1.
\end{eqnarray}
\end{prop}
Note that for $m=2$ the problem is trivial, the party with XOR function can always recover the other input.

\begin{proof}
We start by the converse. 
To show that $\rho (\emptyset,\dots,\emptyset,X_1\oplus \cdots \oplus X_m)\geq 1$, it is enough to show $\rho (\emptyset, \emptyset, X_1\oplus X_2 \oplus X_3)\geq 1$, since increasing the number of parties only increases the randomness required by the protocol.

Clearly, since $P_3$ has input $X_3$, we have $\rho (\emptyset, \emptyset, X_1\oplus X_2 \oplus X_3)=\rho (\emptyset, \emptyset, X_1\oplus X_2)$. 
Denote by $A_1$ all the information that was exchanged between $P_1$ and $P_3$ throughout the protocol, and 
denote by $A_2$ all the information that was exchanged between $P_2$ and $P_3$ throughout the protocol. 

From the accuracy requirement, we have ${H(X_1\oplus X_2|A_1,A_2)=0}$,
in addition to this, from the security requirement on $P_3$, we have 
${H(X_1|A_1,A_2)=1\text{ and }H(X_2|A_1,A_2)=1}$,
which is equivalent to 
\begin{eqnarray}
H(X_1,X_2|A_1,A_2)&=&1\label{xor:accuracy}\\
H(X_1|A_1,A_2)&=&1\label{xor:p3 sec1}\\
H(X_2|A_1,A_2)&=&1.\label{xor:p3 sec2}
\end{eqnarray}
From the security requirement on $P_1$ and $P_2$, we have 
\begin{eqnarray}
H(X_2|A_1,X_1)=1\text{ and }H(X_1|A_2,X_2)=1.\label{xor:p2 sec}
\end{eqnarray}
Finally, from the independence of the inputs
\begin{eqnarray}
H(X_1|X_2)=1\text{ and }H(X_2|X_1)=1.\label{xor:x1 ind}
\end{eqnarray}
Since $A_1$ and $A_2$ are only a part of all information transmitted and received,
\begin{align}
\rho(\emptyset,\emptyset,X_1 \oplus X_2 \oplus X_3) \geq H(A_1,A_2|X_1,X_2).
\end{align}
We now show that the last term is more than 1. We have
\begin{eqnarray}
H(A_1,A_2|X_1,X_2)&\geq& H(A_2|X_1,X_2)\nonumber\\
&=&H(A_2|X_2)\label{xor:pf 1}\\
&=&H(A_2)\label{xor:pf 2}\\
&\geq& 1.\label{xor:pf 3}
\end{eqnarray}
(\ref{xor:pf 1}) follows by (\ref{xor:p2 sec}) and (\ref{xor:x1 ind}), because 
\begin{eqnarray}
H(X_1|A_2,X_2)&=&H(X_1|X_2)\nonumber\\
\Longleftrightarrow \,\, H(A_2|X_1,X_2)&=&H(A_2|X_2),\nonumber
\end{eqnarray}
and (\ref{xor:pf 2}) follows by (\ref{xor:p3 sec2}) because $X_2$ and $(A_1,A_2)$ are independent, so in particular, $X_2$ and $A_2$ are independent. With the same argument, we notice that \eqref{xor:p3 sec1} implies $H(X_1|A_1)=1$.
Finally, (\ref{xor:pf 3}) follows by (\ref{xor:accuracy}) and (\ref{xor:p3 sec1}) 
because 
\begin{eqnarray}
H(X_1,X_2|A_1,A_2)&=&H(X_1|A_1)\nonumber\\
\Longleftrightarrow \,\, H(X_2,A_2|X_1,A_1)&=&H(A_2|A_1),\label{xor:pf 4}
\end{eqnarray}
and with (\ref{xor:p3 sec2})
\begin{eqnarray}
H(A_2)&\geq&H(A_2|A_1)\nonumber\\
&=&H(X_2,A_2|X_1,A_1) \nonumber\\
&\geq&H(X_2|X_1,A_1)=1.\nonumber
\end{eqnarray}

We now move to the direct part. 
The achievability of the lower bound is obtained with the following protocol. Consider a multi-party computation protocol for $X_1\oplus \cdots \oplus X_m$. Let $A_1,...,A_{m-1}$ be information that $P_m$ receives during the protocol from $P_1,...,P_{m-1}$, respectively.
\begin{enumerate}
\item $P_1$ draws a random number $Z\in\{0,1\}$ uniformly at random and sends $Z$ to $P_m$, $Y_2=Z\oplus X_1$ to $P_2$.
\item For $k\in\{2,...,m-1\}$, $P_k$ receives $Y_{k}$ from $P_{k-1}$, computes $Y_{k+1}=Y_{k}\oplus X_k=Z\oplus X_1\oplus \cdots \oplus X_k$ and sends $Y_{k+1}$ to $P_{k+1}$.
\item $P_m$ receives $Y_{m}$ from $P_{m-1}$ and computes $Y_{m}\oplus X_m\oplus Z=Z\oplus X_1\oplus \cdots \oplus X_m\oplus Z=X_1\oplus \cdots \oplus X_m$.
\end{enumerate}
In this example, ${A_1=Z}$, ${A_2=\cdots =A_{m-2}=\emptyset}$, and ${A_{m-1}=Z\oplus X_1\oplus \cdots \oplus X_{m-1}}$.
Thus, 
\begin{eqnarray}
&&H(A_1,...,A_{m-1}|X_1,...,X_{m-1})\nonumber\\
&&=H(Z,X_1\oplus \cdots \oplus X_{m-1}|X_1,...,X_{m-1})\nonumber\\
&&=H(Z)=1.\nonumber
\end{eqnarray}
\end{proof}

\subsection{The real summation function}
\begin{prop}
For any $m\geq 3$, 
\begin{eqnarray}
\rho (\emptyset,\dots,\emptyset, X_1+ \cdots +X_m)\leq \log_2(m).
\end{eqnarray}
\end{prop}
Note that for $m=2$ the problem is trivial, the party with the summation can always recover the other input.
\begin{proof}
The upper bound of the randomness is shown by the achievability. Consider a multi-party computation protocol for $X_1+\cdots+X_m$. Let $A_1,...,A_{m-1}$ be information that $P_m$ receives during the protocol from $P_1,...,P_{m-1}$, respectively.
\begin{enumerate}
\item $P_1$ draws a random number $Z\in\{0,1,\dots,m-1\}$ uniformly at random and sends $Z$ to $P_m$, $Y_2=Z+X_1\mod m$ to $P_2$.
\item For $k\in\{2,...,m-1\}$, $P_k$ receives $Y_{k}$ from $P_{k-1}$, computes the addition in modulo $m$
\begin{eqnarray}
Y_{k+1}=Y_{k}+X_k=Z+X_1+\cdots+X_k\nonumber
\end{eqnarray}
and sends $Y_{k+1}$ to $P_{k+1}$.
\item $P_m$ receives $Y_{m}$ from $P_{m-1}$ and computes the addition and subtraction in modulo $m$
\begin{eqnarray}
Y_{m}+X_m-Z=Z+X_1+\cdots+ X_m-Z=X_1+\cdots+X_m.\nonumber
\end{eqnarray}
\end{enumerate}
In this example, ${A_1=Z}$, ${A_2=\cdots =A_{m-2}=\emptyset}$, and ${A_{m-1}=Z+X_1+\cdots+X_{m-1}}$.
Thus, 
\begin{eqnarray}
&&H(A_1,\dots,A_{m-1}|X_1,...,X_{m-1})\nonumber\\
&&=H(Z,X_1+\cdots+X_{m-1}|X_1,\dots,X_{m-1})\nonumber\\
&&=H(Z)\nonumber =\log_2 m.\nonumber
\end{eqnarray}
\end{proof}

\begin{conjecture}\label{conj1}
For any $m\geq 3$, 
\begin{eqnarray}
\rho (\emptyset,\dots,\emptyset, X_1+ \cdots +X_m)=\log_2(m).
\end{eqnarray}
\end{conjecture}
While showing that $\log_2(m)$ is necessary is not established, we believe that a logarithmic bound in $m$ can be obtained with similar argument as for the XOR function. In particular, this can be written as an inequality over entropic vectors, for which Shanon-type inequalities may or may not suffice. 

\begin{remark}\label{remark_mult}
While this paper focuses on summation functions, similar methods can be used for multiplications. Consider for example the case where three parties wish to compute $(X_1X_2, X_1X_2,X_1X_2)$, i.e., the product of the first two parties' bits. This can be achieved with a protocol requiring $4$ bits of randomness. One possibility is two break each number into three shares, two of which being uniformly distributed, i.e., $X_1=X_1(1)+X_1(2)+X_1(3)$ and $X_2=X_2(1)+X_2(2)+X_2(3)$, and requiring party $P_1$ and $P_2$ to exchange all the bits $X_i(j)$ for $i,j \in \{1,2\}$ and to provide the bit $X_1(3), X_2(3)$ to party 3. Then each party has a component of the product $X_1 X_2$ which can be transmitted to $P_3$ for the function computation. 
\end{remark}

%
%
\section{Discrete memoryless secure multiparty computations}
We now define a probabilistic model for the parties' inputs, and leverage the correlations among these inputs to obtain protocols which are secure with high probability in the limit of large sequences, without requiring private communication channels between all parties. 


\begin{definition}
Let $n \geq 1$ and $(X_{(1)}^n, \dots, X_{(m)}^n)$ be i.i.d.\ sequences with a joint distribution $\mu$ on $\F_2^m$. 
Let $P_1,\dots,P_m$ be $m$ parties, where party $P_i$ possesses the input sequence $X^n_{(i)}$ and the distribution $\mu$, for $i \in [m]$. 

We are now interested in sequences of deterministic protocols, defined on the HCN model without item 2), where parties exchange only deterministic functions of their inputs (no action a)). 
A sequence of deterministic protocols $\{\Pi_n\}_{n \geq 1}$ computes asymptotically accurately and privately the deterministic and discrete function sequence 
$\{(f_{(1)}^n,\dots,f_{(m)}^n)\}_{n \geq 1}$ if 
\begin{itemize}
\item {\bf [Asymptotic accuracy]} Each party $P_i$ can compute $f^n_{(i)}(X_{(1)}^n, \dots, X_{(m)}^n)$ with a vanishing error probability, i.e., from the view of the protocol $Y_{(i)}(\Pi_n)$ and its input $X_{(i)}^n$, party $i$ can compute an estimate $\hat{f}_{(i)}^n$ such that 
\begin{align}
\pp \{ \hat{f}_{(i)}^n \neq  f^n_{(i)}(X_{(1)}^n, \dots, X_{(m)}^n)\} \to 0, \quad \text{as } n \to \infty,
\end{align}
\item  {\bf [Asymptotic security]} Each party $P_i$ cannot recover the input of another party, i.e., for any $j \neq i$, there is no function $\hat{X}_{(j)}^n$ of $Y_{(i)}(\Pi_n)$ and $X_{(i)}^n$ such that  
\begin{align}
\lim_{n \to \infty} \pp \{ \hat{X}_{(j)}^n \neq X_{(j)}^n \} \to  0.
\end{align}
\end{itemize}
\end{definition}
Note that the above definition of security is weaker than its counter-part in the one-shot setting by more than just its asymptotic nature: it is not forbidden to just leak some information, but to actually recover an input sequence.  

Given a set of functions $\{f^n_{(1)}, \dots, f^n_{(m)}\}$, our goal is to study for which distributions $\mu$ on $\F_2^m$ it is possible to obtain a protocol computing the functions accurately and securely in the above asymptotic sense. 

\subsection{The XOR function}

In this section, we introduce an asymptotically accurate and secure protocol for the modulo-2 sum of three parties inputs. Namely $f^n_{(i)}=X^n\oplus Y^n \oplus Z^n$ for $i=1,2,3$. 

\begin{definition}
Let $X$ and $Y$ be binary random variables with a joint distribution $\mu$ on $\F_2^2$. We call the distribution \emph{additively-correlated} if
\begin{eqnarray}
H(X,Y)-2H(X\oplus Y)>0. \label{eqn:strong correlation}
\end{eqnarray}
Let $X,Y,Z$ be binary random variables with a joint distribution $\mu$ on $\F_2^3$. We call the distribution \emph{additively-correlated} if
at least one pair of the random variables is additively-correlated. 
\end{definition}
For example, $X\sim Ber(0.5)$, $Z\sim Ber(p)$, and $Y=X\oplus Z$, where $p<0.5$ satisfies \eqref{eqn:strong correlation}.

\begin{prop}\label{polar-prop}
Let $n \geq 1$ and $(X^n, Y^n, Z^n)$ be i.i.d.\ sequences with a joint distribution $\mu$ on $\F_2^m$ which is pariwise-additively-correlated. 
Then the ASP protocol defined below allows to compute asymptotically accurately and securely the function $X^n \oplus Y^n \oplus Z^n$.  
\end{prop}
This provides an achievability result. 

\begin{remark}
The ASP protocol is based on polar codes. The linearity of the code is crucial to compute the XOR function. The protocol could probably be adapted with other linear codes, such as random linear codes, however, polar codes provide in addition a low-complexity protocol, and are also insightful as a proof technique. 
\end{remark}

%

We next recall the source polarization results and then describe the protocol. 

\subsection{Preliminaries on polar codes}

For $n$ a power of $2$, define $G_n=\begin{pmatrix} 1 & 0 \\ 1& 1 \end{pmatrix}^{\otimes \log _2 (n)}$, where $A^{\otimes k}$ denotes the matrix obtained by taking $k$ Kronecker products of matrix $A$ with itself.
\begin{thm}\cite{ari3}
\label{thm:1}
Let $X^n=[X_1,...,X_n]\overset{iid}{\sim}$ Bernoulli(p), where $n$ is a power of $2$, and let $\tilde{X}^n=X^nG_n$. Then, for any $\epsilon \in (0,1/2]$,  
\begin{eqnarray}
|\{j\in[n]:H(\tilde{X}_j|\tilde{X}^{j-1})\in (\epsilon,1-\epsilon)\}|=o(n),\nonumber
\end{eqnarray} 
where $H(\tilde{X}_j|\tilde{X}^{j-1})$ represents the conditional Shannon entropy of $\tilde{X}_j$ given $\tilde{X}^{j-1}=[\tilde{X}_1,...,\tilde{X}_{j-1}]$. The above still holds if $\e=O(2^{-n^{\beta}})$, $\beta <1/2$. 
\end{thm}
Theorem \ref{thm:1} says that, except for a vanishing fraction, all conditional entropies $H(\tilde{X}_j|\tilde{X}^{j-1})$ tend to either $0$ or $1$. Also, notice that since $G_n$ is invertible, hence $nH(p)=H(X^n)=H(\tilde{X}^n)$, 
and defining 
\begin{eqnarray}
R_{\epsilon,n}(X):=\{i\in[n]:H(\tilde{X}_i|\tilde{X}^{i-1})\geq \epsilon\},\label{eqn:polarset}
\end{eqnarray}
we have
\begin{eqnarray}
&&\frac{1}{n}| R_{\epsilon,n}(X) | \rightarrow H(p), \label{eqn:opt rate}
\end{eqnarray}
where $H(p)$ is the entropy of the Bernoulli(p) distribution.
Note that $R_{\e,n}(X)$ should be written $R_{\e,n}(p)$: it is not a function of a random variable but $X$, but here $X$ stands for the marginal distribution of $X^n$. This notation will be handy below. Since $R_{\e,n}(X)$ contains all the non-deterministic components of $\tilde{X}^n$, it is possible to reconstruct $\tilde{X}^n$ from $\tilde{X}[R_{\e,n}(X)]$ with a vanishing probability of error. This requires setting $\e$ small enough, in particular one can chose $\e=O(2^{-n^{\beta}})$, $\beta <1/2$. With polar codes, one can addition obtain an efficient decoding algorithm which runs in $O(n \log(n))$. 



\subsection{The asymptotically secure polar (ASP) protocol}
All the parties know $\mu$ and set $\e=\e_n=2^{-n^{0.49}}$. Since $\mu$ is additively-correlated, assume w.l.o.g\ that the inputs of $P_1$ and $P_2$ are additively-correlated. 
\begin{enumerate}
\item Inputs at time 0: \\party 1: $X^n$, party 2: $Y^n$, party 3: $Z^n$.
\item At time 1:
\begin{itemize}
 \item $P_1$ computes $\tilde{X}^n=X^nG_n$ and $R_{\epsilon,n}(X\oplus Y)$,  
 \item $P_2$ computes $\tilde{Y}^n=Y^nG_n$ and $R_{\epsilon,n}(X\oplus Y)$.
\end{itemize} 
\item At time 2:
\begin{itemize}
 \item $P_1$ sends $\tilde{X}[R_{\epsilon,n}(X\oplus Y)]$ to $P_3$,
 \item $P_2$ sends $\tilde{Y}[R_{\epsilon,n}(X\oplus Y)]$ to $P_3$.
\end{itemize} 
\item At time 3:  
 $P_3$ computes $\tilde{X}[R_{\epsilon,n}(X\oplus Y)]\oplus \tilde{Y}[R_{\epsilon,n}(X\oplus Y)]=(\tilde{X}\oplus \tilde{Y})[R_{\epsilon,n}(X\oplus Y)]$ and decodes $\tilde{X}^n\oplus \tilde{Y}^n$ from $(\tilde{X}\oplus \tilde{Y})[R_{\epsilon,n}(X\oplus Y)]$ using the polar decoding algorithm in \cite{ari3}. Let $\hat{\tilde{X}}^n\oplus \hat{\tilde{Y}}^n$ be the decoded vector. $P_3$ computes $\hat{X}^n\oplus \hat{Y}^n=(\hat{\tilde{X}}^n\oplus \hat{\tilde{Y}}^n)G_n^{-1}$. 
 \item At time 4: $P_3$ sends $\hat{X}^n\oplus \hat{Y}^n + Z^n$ to $P_1$ and $P_2$.
\end{enumerate}

\subsection{Proof of Proposition \ref{polar-prop}} 
\begin{lemma}
The SPC protocol is \emph{asymptotically accurate}.
\end{lemma}

\begin{proof}
Since $Z^n$ is not encoded during the protocol, it is enough to prove that
\begin{eqnarray}
Pr(\hat{X}^n\oplus \hat{Y}^n\neq X^n\oplus Y^n)\xrightarrow[n\rightarrow \infty]{} 0.
\end{eqnarray}
This is a direct application of Theorem \ref{thm:1}, as in \cite{ari3}.   
\end{proof}

\begin{lemma}
The SPC is \emph{asymptotically secure}.
\end{lemma}

\begin{proof}
Since $P_1$ receives only $X^n\oplus Y^n\oplus Z^n$ during the protocol, it is clear that it cannot estimate $Y^n$ or $Z^n$ with a vanishing error probabilities. Similarly, $P_2$ cannot estimate $X^n$ or $Z^n$ with a vanishing error probability. Therefore, to prove that the protocol 2 is asymptotical secure for all parties, it is enough to prove that for $P_3$,
\begin{eqnarray}
Pr(\hat{X}^n\neq X^n) \gg 0\label{eqn:security1}
\end{eqnarray}
where $a_n \gg 0$ means $\liminf_{n \to \infty} a_n >0$, 
and
\begin{eqnarray}
Pr(\hat{Y}^n\neq Y^n)\gg0,\label{eqn:security2}
\end{eqnarray}
where $\hat{X}^n$ and $\hat{Y}^n$ are $P_3$'s estimations of $X^n$ and $Y^n$ given $\tilde{X}[R_{\epsilon,n}(X\oplus Y)]$ and $\tilde{Y}[R_{\epsilon,n}(X\oplus Y)]$. 

During the protocol, $P_3$ receives $\tilde{X}[R_{\epsilon,n}(X\oplus Y)]$ and $\tilde{Y}[R_{\epsilon,n}(X\oplus Y)]$ and knows $X^n\oplus Y^n$ with the vanishing error probability by Theorem \ref{thm:1}. Thus, for $P_3$, knowing $X^n$ with the vanishing error probability guarantees recovery of $Y^n=X^n\oplus (X^n\oplus Y^n)$, and vice verse. Then,
\begin{eqnarray}
&& Pr(\hat{X}^n\neq X^n)\rightarrow 0 \cup Pr(\hat{Y}^n\neq Y^n)\rightarrow 0\nonumber\\
&&\Leftrightarrow Pr(\hat{X}^n,\hat{Y}^n\neq X^n,Y^n)\rightarrow 0.\label{polar:error equiv}
\end{eqnarray}
Therefore, it is enough to show $Pr(\hat{X}^n,\hat{Y}^n\neq X^n,Y^n)\gg0$. 
Notice that 
\begin{eqnarray}
X^n, Y^n\rightarrow \tilde{X}[R_{\epsilon,n}(X\oplus Y)],\tilde{Y}[R_{\epsilon,n}(X\oplus Y)]\rightarrow \hat{X}^n,\hat{Y}^n\nonumber
\end{eqnarray}
forms a Markov chain. Then, by Fano's inequality,
\begin{align}
&Pr(\hat{X}^n,\hat{Y}^n\neq X^n,Y^n)\nonumber
\\&\geq \frac{H(X^n, Y^n|\tilde{X}[R_{\epsilon,n}(X\oplus Y)],\tilde{Y}[R_{\epsilon,n}(X\oplus Y)]-1}{2n}\nonumber
\\&\geq \frac{H(X^n, Y^n)-H(\tilde{X}[R_{\epsilon,n}(X\oplus Y)],\tilde{Y}[R_{\epsilon,n}(X\oplus Y)])-1}{2n}\nonumber
\\&\geq \frac{H(X^n, Y^n)-2H(\tilde{X}[R_{\epsilon,n}(X\oplus Y)])-1}{2n}\nonumber
\\&\geq \frac{H(X^n, Y^n)-2|R_{\epsilon,n}(X\oplus Y)|-1}{2n}\nonumber
\\&= \frac{nH(X, Y)-2nH(X\oplus Y)-1}{2n}\label{entropy bound}
\end{align}
where (\ref{entropy bound}) follows from \eqref{eqn:opt rate}. 
Then, since $X^n$ and $Y^n$ are \emph{additively-correlated},
\begin{eqnarray}
\liminf_{n \to \infty} Pr(\hat{X}^n,\hat{Y}^n\neq X^n,Y^n)>0.
\end{eqnarray}
\end{proof}
\begin{remark}
Note that since the ASP protocol is asymptotically secure, 
it must be that 
$H(X,Y)-2H(X\oplus Y)>0$ implies $H(X)-H(X\oplus Y)>0$ and $H(Y)-H(X\oplus Y)>0$. 
The reason for that is due to the nested property of source polar codes, see for example \cite{abbe-corr}, which implies that if $H(X+Y) > H(X)$, then $R_{\e,n}(X+Y)$ contains $R_{\e,n}(X)$ and hence observing $\tilde{X}^n[R_{\e,n}(X+Y)]$ allows to decode $X^n$ correctly. 
In fact the above implication is true since  
\begin{align}
&H(X,Y)-2H(X\oplus Y) >0\nonumber\\
&\Longleftrightarrow H(X|X+Y) > H(X\oplus Y)
\end{align}
and $H(X|X+Y) \leq H(X)$. 

\end{remark}

%
%

\section{Open problems}
Concerning the first part, it would interesting to set conjecture \ref{conj1}. We believe that a logarithmic bound can be obtained with the approach of this paper. 
As mentioned in Remark \ref{remark_mult}, it is possible to obtain achievability results on $\rho$ for multivariate polynomials. The scaling of the $\rho$ can then be analyzed.    
For the second part, it would be interesting to establish converse results for XOR function, and any result for other type of functions, starting perhaps with the real-addition. 

%

\section*{Acknowledgements}
We thank R.~Chou and M.~Bloch for stimulating discussions.  

\bibliographystyle{amsplain}
\bibliography{crypto}


\end{document}